\documentclass[11pt]{article}
\usepackage[section]{algorithm}
\usepackage[noend]{algorithmic}
\usepackage{graphicx}
\usepackage{amsfonts}
\usepackage{amssymb}
\usepackage{amsthm}
\usepackage{textcomp}
\usepackage{enumerate}
\usepackage[shortlabels]{enumitem}
\usepackage{subfigure}
\usepackage{hyperref}
\usepackage{fullpage}

\newtheorem{definition}{Definition}[section]
\newtheorem{theorem}[definition]{Theorem}
\newtheorem{lemma}[definition]{Lemma}
\newtheorem{corollary}[definition]{Corollary}

\title{Terminating cases of flooding} 

  \author{Walter Hussak \\  Computer Science, \\ Loughborough University, \\ United Kingdom\\ W.Hussak@lboro.ac.uk
   \and Amitabh Trehan \\ Computer Science, \\  Loughborough University, \\ United Kingdom \\ Amitabh.Trehaan@gmail.com}
   \date{}

\begin{document}

\maketitle

\begin{abstract}
Basic synchronous flooding proceeds in rounds. Given a finite undirected (network) graph $G$, a set of sources $I \subseteq G$ initiate flooding in the first round by every node in $I$ sending the same message to all of its neighbours. In each subsequent round, nodes send the message to all of their neighbours from which they did not receive the message in the previous round. Flooding terminates when no node in $G$ sends a message in a round.  The question of termination has not been settled -  rather, non-termination is implicitly assumed to be possible.

We show that flooding terminates on every finite graph. In the case of a single source $g_0$, flooding  terminates in $e$ rounds  if $G$ is bipartite and $j$ rounds with $e < j \leq e+d+1$ otherwise, where $e$ and $d$ are the eccentricity of $g_0$ and diameter of $G$ respectively. For communication/broadcast to all nodes, this is asymptotically time optimal and obviates the need for construction and maintenance of spanning structures. We extend to dynamic flooding initiated in multiple rounds with possibly multiple messages. The cases where a node only sends a message to neighbours from which it did not receive {\it  any} message in the previous round, and where a node sends some highest ranked message to all neighbours from which it did not receive {\it that} message in the previous round, both terminate. All these cases also hold if the network graph loses edges over time. Non-terminating cases include asynchronous flooding, flooding where messages have fixed delays at edges, cases of multiple-message flooding and cases where the network graph acquires edges over time.
\end{abstract}


\section{Introduction}


Flooding is one of the most fundamental of all graph/network algorithms - for example, to achieve \emph{broadcast} in networks~\cite{Attiya-WelchBook,pelegbook}. The algorithm that is usually implemented is: (i) initiator nodes send a message to all their neighbours, (ii) on first receipt, receiving nodes send the message to every neighbour from which they did not receive and (iii) on subsequent receipts, the receiving node does not send the message to any neighbour. This algorithm checks if the message has been received before, to ensure termination. It requires each node records a previous event in its state, making it a \emph{stateful} protocol.  We consider the pure version of flooding, i.e. algorithms that are `amnesiac' in the sense that they hold no record or memory of any event beyond the immediate round. Thus, implemented algorithms that are {\it stateless} are amnesiac. The basic pure flooding algorithm is as follows.
\begin{definition}{\textbf{Synchronous amnesiac flooding algorithm.}}
Let $(G,E)$ be an undirected graph, with vertices $G$ and edges $E$ (representing a network where the vertices represent the nodes of the network and edges represent the connections between the nodes).  Computation proceeds in synchronous `rounds' where each round consists of nodes receiving messages being sent from neighbours. A receiving node sends messages to some neighbours in the next round. No messages are lost in transit. The algorithm is defined by the following conditions.
\begin{itemize}
\item[(i)]
All  nodes from a subset of sources or `initial nodes' $I \subseteq G$  send a message $M$ to all of their neighbours in round 1. 
\item[(ii)]
In subsequent rounds, every node that received $M$ from a neighbour in the previous round, sends $M$ to all, and only those, nodes from which it did not receive $M$. Flooding \emph{terminates} when $M$ is no longer sent to any node in the network.
\end{itemize}
\end{definition}

Hitherto, it was not known whether a stateless flooding algorithm exists that terminates on every graph. In this paper, we show that basic amnesiac flooding is itself such an algorithm. The termination times for a single source $g_0$ are $e$ rounds  if $G$ is bipartite and $j$ rounds with $e < j \leq e+d+1$ otherwise, where $e$ and $d$ are the eccentricity of $g_0$ and diameter of $G$ respectively. These times are asymptotically time optimal. The difference in times between bipartite and non-bipartite graphs means that a stateless algorithm can approximate the topology or diameter/eccentricities of a network in certain cases.  
 
 Gopal, Gopal and Kutten~\cite{GopalGK99} highlighted the memory requirement inherent in stateful flooding and its overhead on fast networks. For example, when multiple messages are being flooded, nodes need to keep a record of the messages they have seen and search through their data structure whenever a message arrives before deciding if it can be further flooded. This is also a limitation for low memory devices, e.g. sensor networks. We show that not only amnesiac flooding of the same message in a single round,  but also amnesiac flooding of multiple messages initiated in multiple rounds, where at most one message is sent along any given edge each round, terminates.
 
 One advantage of  the stateful variant of flooding is that it can be used to construct spanning structures, such as spanning trees, which can be subsequently re-used for efficient communication~\cite{Lynchbook, Attiya-WelchBook, pelegbook,GerardTelDistributedAlgosBook}. It is not easy to maintain these structures in dynamic conditions, i.e. when edges/nodes are being added or removed. We show that even when edges are being arbitrarily or adversarially removed from the graph, amnesiac flooding with the correct forwarding rules will still terminate. Under some other forwarding rules, termination is not guaranteed. Termination is also not guaranteed in asynchronous models where the delivery time of messages on edges can vary - we show this with a non-terminating delivery schedule over the edges and even with a single edge with delayed delivery in some pathological cases.  Note that amnesiac flooding uses rules that avoid forwarding to the most recently chosen nodes - similar ideas have been used before in distributed computing, e.g. in social networks~\cite{Doerr-Social-EuroComb11} and broadcasting~\cite{Elsasser-Memory-SODA08}.

There are numerous applications of flooding as a distributed algorithm. It is often used to set up spanning trees\cite{Lynchbook, Attiya-WelchBook, pelegbook, GerardTelDistributedAlgosBook} and solving \emph{leader election}. For example, the time-optimal leader election of Peleg~\cite{PelegLE-JPDC-90} utilises continuous flooding of node $ID$s to determine the node with the highest $ID$ while using a condition on the node $ID$s to achieve and detect termination. Kutten et al~\cite{KuttenPP0T-JACM-LE-15}, on the other hand, use controlled stateful flooding and resulting spanning tree for efficient time and message leader election. Processes such as random walks~\cite{drw1,drw2,mihail-p2p,KuttenPP0T-JACM-LE-15,KuttenSublinear-TCS-2015} and its deterministic variant Rotor-Router (or Propp) machine~\cite{KosowskiP19,cooper_spencer_2006} can also be seen as restricted variants of flooding. Finally, in some models such as population protocols, the low memory makes termination very difficult to achieve leading to research that tries to provide termination e.g.~\cite{Michail-SpirakisPP2015}.

Flooding-based algorithms (or flooding protocols) appear in areas ranging from GPUs, high performance, shared memory and parallel computing, to mobile ad-hoc networks (MANETs), Mesh Networks, Complex Networks etc~\cite{tanenbaum2011computer}. Rahman et al~\cite{Rahman04controlledflooding} show that  flooding can even be adopted as a reliable and efficient routing scheme, comparable to  sophisticated point-to-point forwarding schemes, in some ad-hoc wireless mobile network applications. Adamek et al~\cite{AdamekNRT-SRDS17} and Karp and Kung~\cite{Karp-GPSR-Mobicom2000} further offer stateless (flooding and other) solutions for routing.

 
 

This paper is structured as follows. In section 2, we prove that, for the simple case of a single message being flooded from a set of initial nodes in round $0$, all nodes receive the message at most twice and hence amnesiac flooding terminates (Theorem 2.2). This result is used to derive sharp bounds for the number of rounds it takes for flooding to terminate, in section 3 (Theorem 3.9). Amnesiac flooding is extended in section 4 to dynamic cases where multiple messages are flooded in any round. The flooding algorithms are categorized by whether a node sends a chosen received message $M_h$ to neighbours from which it did not receive any message in the previous round, or whether the node sends $M_h$ to all neighbours from which it did not receive $M_h$.  We show that, in the first case, and also the second case if messages are ranked, no flooded message is received more than twice and therefore flooding terminates if the number of messages is finite (Theorems 4.4 and 4.8). In section 5 we present  non-terminating cases of amnesiac flooding for both single-message and multiple-message flooding. Conclusions are drawn in section 6.

Sections 2 and 3 of this paper have essentially appeared earlier as conference publications at STACS 2020~\cite{HussakT-STACS20} and PODC 2019~\cite{HussakT19}, and on Arxiv~\cite{HussakTArxiv19}. The proof of time to termination was presented for the case of a single source, but the same proof also provides the time calculation for the multi-source case simply by replacing distances $d(g_0,g)$ of a node $g$ from the initial node $g_0$  by distances $d(I,g)$ of $g$ from the set of initial nodes $I$. Since those publications, Turau \cite{turau1,Turau-Stateless-SIROCCO20} has worked on finding optimal sets of sources of a given size that minimize termination time for multi-source flooding.  The parts of this paper from section 4 onwards on dynamic flooding, which require the use of a more powerful proof method, have not been published previously.

\section{Termination}
Throughout this paper, when we refer to `flooding' we will mean either the amnesiac flooding of Definition 1.1 or another amnesiac version of flooding defined in context in the relevant section. In this section we prove that (amnesiac) flooding initiated from any set of initial nodes, all at once, terminates.
\begin{definition}Let $G$ be the graph and $I \subseteq G$ a set of initial nodes. The {\it round-sets} $R_0, R_1, \ldots $ are defined as:
\[
\begin{array}{lll}
R_0 & \hbox{\it is the set $I$ containing the initial nodes,} & \; \\
R_i & \hbox{\it is the set of nodes which receive a message in round i} & (i\geq 1). \\
\end{array}
\]
\end{definition}Clearly, if $R_j = \emptyset$ for some $j \geq 0$, then $R_i = \emptyset$ for all $i \geq j$.
\begin{theorem}
Any node $g \in G$ is contained in at most two distinct round-sets.
\end{theorem}
\begin{proof}
Define ${\mathcal R}$ to be the set of finite sequences of consecutive round-sets of the form:
\begin{equation}
\underline{R} = 
R_s, \ldots, R_{s+d} \;\;\; \hbox{ {\it where} $s\geq 0, \;d>0$, {\it and} $R_s \cap R_{s+d} \neq \emptyset$ }.
\end{equation}In (1), $s$ is the {\it start-point} $s(\underline{R})$ and $d$ is the {\it duration} $d(\underline{R})$ of $\underline{R}$. Note that, a node $g \in G$ belonging to $R_s$ and $R_{s+d}$  may also belong to other $R_i$ in (1).  If a node $g \in G$ occurs in three different round-sets $R_{i_1}$, $R_{i_2}$ and $R_{i_3}$, then the duration between $R_{i_1}$ and $R_{i_2}$, $R_{i_2}$ and $R_{i_3}$, or $R_{i_1}$ and $R_{i_3}$ will be even. Consider the subset ${\mathcal R}^e$ of ${\mathcal R}$ of sequences of the form (1) where $d$ is even. To prove that no node is in three round-sets, it suffices to prove that ${\mathcal R}^e$ is empty. 

We assume that ${\mathcal R}^e$ is non-empty and derive a contradiction. Let ${\mathcal R}^e_{\hat{d}}$ be the subset of ${\mathcal R}^e$ comprising sequences of minimum (even) duration $\hat{d}$, i.e.
\begin{equation}
{\mathcal R}_{\hat{d}}^e = \{ \underline{R} \in {\mathcal R}^e \;\; | \;\; \forall \; \underline{R}'\in {\mathcal R}^e . \;\; d(\underline{R}') \geq d(\underline{R})=\hat{d} \}
\end{equation}Clearly, if ${\mathcal R}^e$ is non-empty then so is ${\mathcal R}^e_{\hat{d}}$. Let $\underline{R}^* \in {\mathcal R}^e_{\hat{d}}$ be the sequence with earliest start-point $\hat{s}$,  i.e.
\begin{equation}
\underline{R}^* = R_{\hat{s}}, \ldots , R_{\hat{s}+\hat{d}}
\end{equation}where
\begin{equation}
 \forall \; \underline{R}'\in {\mathcal R}_{\hat{d}}^e\;.\; s(\underline{R}') \geq s(\underline{R}^*) = \hat{s} 
\end{equation}By (1), there exists $g \in R_{\hat{s}} \cap R_{\hat{s}+\hat{d}}$. Choose node $g'$ which sends a message to $g$ in round $\hat{s}+\hat{d}$. As $g'$ is a neighbour of $g$, either $g'$ sends a message to $g$ in round $\hat{s}$ or $g$ sends a message to $g'$ in round $\hat{s}+1$. We show that each of these cases leads to a contradiction.

$\;$

\noindent {\it Case (i) $g'$ sends a message to $g$ in round $\hat{s}$}

\noindent In this case,  there must be a round $\hat{s}-1$ which is either round 0 and $g'$ is an initial node in $I$, or $g'$ received a message in round $\hat{s}-1$.  Thus, the sequence 
\begin{equation}
\underline{R}^{*'} = R_{\hat{s}-1},R_{\hat{s}},\ldots, R_{\hat{s}+\hat{d}-1}\;\;\; \hbox{ {\it where}}\; g' \in R_{\hat{s}-1} \cap R_{\hat{s}+\hat{d}-1} 
\end{equation}has $d( \underline{R}^{*'})=(\hat{s}+\hat{d}-1)-(\hat{s}-1) =\hat{d}$ which is even and so $\underline{R}^{*'} \in {\mathcal R}_{\hat{d}}^e$. As $\underline{R}^{*'} \in {\mathcal R}_{\hat{d}}^e$,  by (4)
\begin{equation}
s(\underline{R}^{*'}) \geq s(\underline{R}^*)
\end{equation}But, from (5), $s(\underline{R}^{*'}) = \hat{s}-1$ and, from (4), $ s(\underline{R}^*)=\hat{s}$. Thus, by (6),
\[
\hat{s}-1 = s(\underline{R}^{*'}) \geq s(\underline{R}^*) = \hat{s}
\]which is a contradiction.

$\;$

\noindent {\it Case (ii) $g$ sends a message to $g'$ in round $\hat{s}+1$}

\noindent  By the definition of ${\mathcal R}^e$, the smallest possible value of $\hat{d}$ is 2. However, it is not possible to have $\hat{d}=2$ in this case as then
\[
\underline{R}^* = R_{\hat{s}}, R_{\hat{s}+1}, R_{\hat{s}+2}
\]and $g$ sends a message to $g'$ in round $\hat{s}+1$. As we chose $g'$ to be such that $g'$ sends a message to $g$ in round $\hat{s}+\hat{d}=\hat{s}+2$  this cannot happen because $g$ cannot send a message to $g'$ and $g'$ to $g$ in consecutive rounds by the definition of rounds. So,
\[
\underline{R}^* = R_{\hat{s}}, R_{\hat{s}+1},\ldots, R_{\hat{s}+\hat{d}-1}, R_{\hat{s}+\hat{d}}
\] where $\hat{s}+1<\hat{s}+\hat{d}-1$. Consider the sequence
\begin{equation}
\underline{R}^{*''} =  R_{\hat{s}+1},\ldots, R_{\hat{s}+\hat{d}-1}
\end{equation}As $g'$ receives a message from $g$ in round $\hat{s}+1$ and $g'$ sends a message to $g$ in round $\hat{s}+\hat{d}$, it is clear that $g' \in 
 R_{\hat{s}+1} \cap R_{\hat{s}+\hat{d}-1}$. Thus, $\underline{R}^{*''} \in {\mathcal R}$. As $\hat{d}$ is even, so is $(\hat{s}+\hat{d}-1)-(\hat{s}+1)=\hat{d}-2$ and therefore  $\underline{R}^{*''} \in {\mathcal R}^e$. Now, $\underline{R}^* \in {\mathcal R}^e_{\hat{d}}$ and so, as $\underline{R}^{*''} \in {\mathcal R}^e$, we have,  by (2), 
\begin{equation}
 d(\underline{R}^{*''}) \geq d(\underline{R}^*)
\end{equation}As $d(\underline{R}^{*''})=\hat{d}-2$ from (7) and $ d(\underline{R}^*) = \hat{d}$ from (3), we have, by (8),
\[
\hat{d}-2 =  d(\underline{R}^{*''}) \geq d(\underline{R}^*) = \hat{d}
\]This contradiction completes the proof.
\end{proof}
\begin{definition}
(For section 3 only - superscripts have a different use in section 4.) Given $g \in G$, we use a superscript 1 to indicate that $g$ belongs to a round-set for the first time, and a superscript 2 to indicate that it belongs to a round-set for the second time, i.e.\[
g^1 \in R_j
\]means that 
\[
g \in R_j \;\;\; {\it and} \;\;\;g \notin R_i \;\;\; \hbox{{\it for all i with}} \;\; 0 \leq i < j.
\]and
\[
g^2 \in R_j
\]means that 
\[
g \in R_j \;\;\; {\it and} \;\;\;g \in R_i \;\;\; \hbox{{\it for some i with}} \;\; 0 \leq i < j.
\]
\end{definition}Theorem 2.2 implies that $R_i = \emptyset$ for $i \geq 2|G|$, where $|G|$ is the number of vertices of $G$, and therefore flooding always terminates. In the next section we give sharp bounds for the number of rounds to termination, in terms of the eccentricity of initial nodes and the diameter of $G$.

\section{Time to termination}
The question of termination of network flooding is non-trivial when cycles are present in $G$. The simple cases when $G$ is an even cycle and when $G$ is an odd cycle, as in figure 1 (where an arrow indicates a message received by a node in the given round), display quite different termination behaviours. 
\begin{figure}[h!]
\centering 
\includegraphics[scale=0.75]{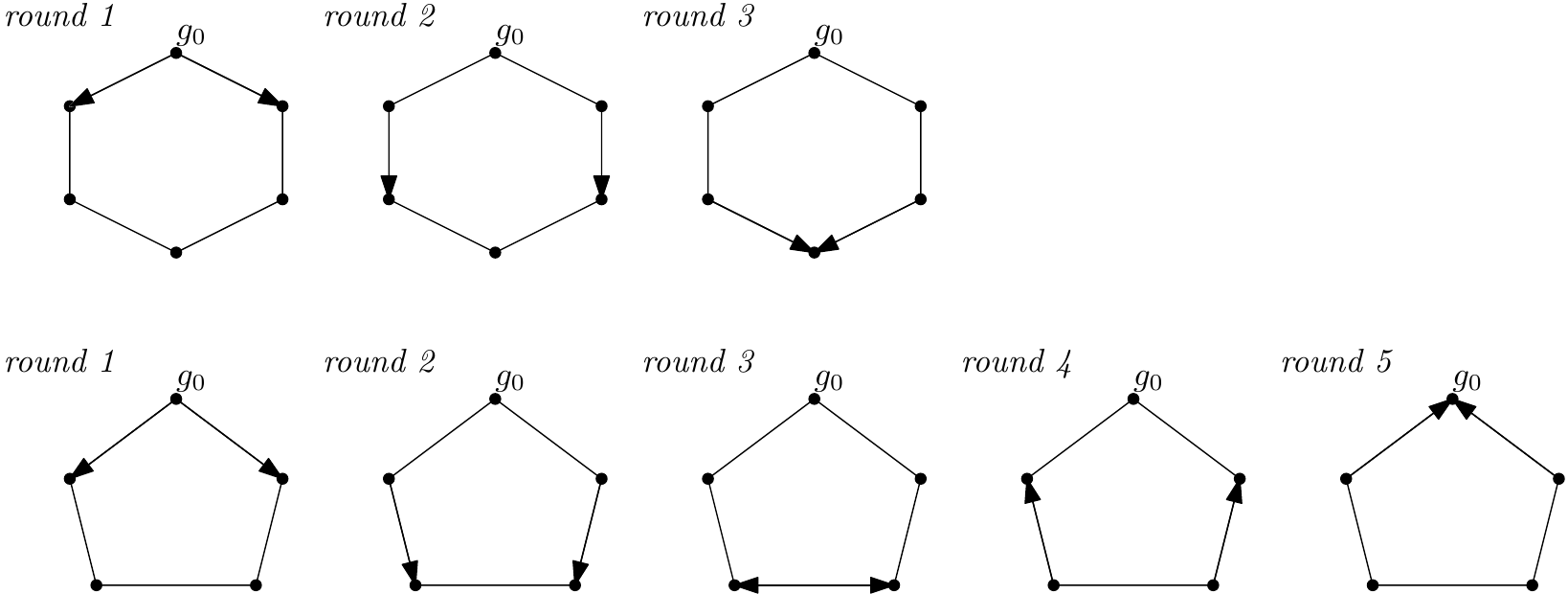} 
\caption{Flooding in even and odd cycles}
\end{figure}The even cycle terminates remotely from the initial node $g_0$, after round $e$ where $e$ is the eccentricity of $g_0$ in $G$. On the other hand, flooding on the odd cycle returns a message to the initial node $g_0$ and terminates after round $2e+1$ resulting in a longer flooding process than the even cycle despite having fewer nodes and a smaller value of $e$. In this section, we show that these observations can be largely generalized to arbitrary graphs. Observe that in the even cycle every node receives the message only once, whereas in the odd cycle every node receives the message twice. We show that, in general graphs, this difference is due to the absence or presence of neighbouring nodes that are equidistant from the set of initial nodes $I$, called `ec nodes' below.
\begin{definition}
Let $(G,E)$ be a graph with vertex set $G$ and edge set $E$ and let $I \subseteq G$ be a set of initial nodes. We will use the following definitions.
\begin{itemize}
\item[(i)]
For each $j  \geq 0$, the distance set $D_j$ will denote the set of points which are a distance $j$ from $I$. i.e.
\[
D_j = \{ g \in G \;   : \; d(I,g) = j \},
\]where, for any $g \in G$,
\[
d(I,g) = {\it min}\{ d(g_0,g) \; : \; g_0 \in I \}
\]and $d$ is the usual distance function in graph $G$.
\item[(ii)]
A node $g \in G $ is an {\it equidistantly-connected node}, abbreviated {\it ec node}, iff it has a neighbour $g'$ equidistant from $I$, i.e.  iff there there exists $g' \in G$ such that $d(I,g)=d(I,g')$ and $\{ g, g' \} \in E$
\item[(iii)]
The {\it eccentricity} of $I$, denoted $e(I)$, is defined as
\[
e(I) = {\it max}\{ d(I,g) \; : \; g \in G \}
\]
\item[(iv)]
The graph $G$ is {\it ec-bipartite} iff $G$ has no ec nodes.
\end{itemize}
\end{definition}\noindent Note that neighbouring initial nodes $g_0,g_0' \in I$ are ec nodes as they are both a distance 0 from $I$. We have the following basic properties of distance sets $D_j$ and $ec$ nodes.
\begin{lemma} Let $G$ be a graph and $I \subseteq G$ be a set of initial nodes.
\begin{itemize}
\item[(i)]
For all $j  \geq 0$ and $i>j$, $D_j \subseteq R_j$ and $R_j \cap D_i = \emptyset$.
\item[(ii)]
For all $j  \geq 0$, $g \in D_j$ and $g' \in D_{j+1}$ such that $g$ and $g'$ are neighbours, $g$ sends a message to $g'$ in round $j+1$, i.e. all nodes at a distance $j$ from $I$ send to all their neighbours which are a distance $j+1$ in round $j+1$.
\item[(iii)]
If $j \geq 0$ and $g \in D_j$ is an $ec$ node, then $g^2 \in R_{j+1}$.
\end{itemize}
\end{lemma}
\begin{proof}
For (i), clearly every node at a distance $j$ from $I$ receives a message in round $j$ and so $D_j \subseteq R_j$. Furthermore, every message received in round $j$ will have travelled at most $j$ edges away from $I$ and so could not have reached a node which is at a distance $i>j$ from $I$. Thus, $R_j \cap D_i = \emptyset$.

For (ii), we note that the only circumstance in which a node $g$ in $D_j$ ($\subseteq R_j$ by (i)) does not send to a neighbour $g'$ in $D_{j+1}$ in round $j+1$ is if $g$ sent a message to $g'$ in round $j$. This would need $g$ to be in the round-set $R_{j-1}$, i.e. $g \in R_{j-1} \cap D_{j+1}$ which contradicts (i) which has $R_{j-1} \cap D_{j+1} = \emptyset$ as $j+1 > j-1$.

For (iii), if $j \geq 0$ and $g \in D_j$ is an $ec$ node, then by Definition 3.1(ii) there is an point $g'$ equidistant from $I$, i.e. $g' \in D_j$ such that $g$ and $g'$ are neighbours. If $j=0$ then $g,g' \in D_0 (= I \subseteq R_0)$ and send a message to each other in round 1 and so $g,g' \in R_1$. This means that $g^2 \in R_1$, as $g$ is an initial node and $g^1 \in R_0$. If $j \geq 1$ then, by (i) of this lemma, $D_j \subseteq R_j$ and so both $g$ and $g'$ receive messages in round $j$.  Also, by (i), neither sends a message in round $j$ as $R_{j-1} \cap D_j = \emptyset$. Thus, $g$ and $g'$ send messages to each other in round $j+1$. As this will be the second time they receive a message we have that $g^2 \in R_{j+1}$
\end{proof}

$\;$

\noindent All nodes in a graph without $ec$ nodes, belong to at most one round-set.
\begin{lemma}
 Let $G$ be a graph and $I \subseteq G$ be a set of initial nodes. Then $G$ has an ec node if and only if $G$ has a node that is in two round-sets.
\end{lemma}
\begin{proof}
Suppose that $G$ has no $ec$ nodes. Assume, on the contrary, that $G$ has nodes that appear in two round-sets.  Let $R_j$ $(j \geq 1)$ be the earliest round which contains a node $g$ such that $g^2 \in R_j$ and $h \in R_{j-1}$ be a neighbour of $g$ which sends to $g$ in round $j$, so that $h^1 \in R_{j-1}$. Then, $h \in D_i$ for some $i \geq 0$ and $h^1 \in R_i$ by Lemma 3.2(i).  Thus, $i=j-1$ and so $g^2 \in R_{i+1}$. As $g$ is a neighbour of $h$, $g \in D_i$, $D_{i+1}$, or $D_{i-1}$. If $g \in D_i$ then $g$ and $h$ are $ec$ nodes contrary to our supposition that $G$ has no $ec$ nodes. If $g \in D_{i+1}$ then $g^1 \in R_{i+1}$ by Lemma 3.2(i), which is contrary to the assertion that $g^2 \in R_j = R_{i+1}$. If $g \in D_{i-1}$ then $g \in R_{i-1}$, by Lemma 3.2(i), and so $g^1 \in R_{i-1}$ as $g^2 \in R_{i+1}$. By Lemma 3.2(ii), $g$ sends to $h$ in round $i=j-1$. This is contrary to $h$ sending to $g$ in round $j$. Thus, our assumption that $G$ has nodes that appear in two round-sets is false.

Conversely, suppose that $G$ has an $ec$ node $g$, $g \in D_j$ ($\subseteq R_j$ by Lemma 3.2(i)) say where $j \geq 1$. Then $g^2 \in R_{j+1}$ by Lemma 3.2(iii).
\end{proof}

$\;$

\noindent In the case of a single initial node, bipartite graphs do not have any $ec$ nodes.
\begin{lemma}
Let $G$ be a graph and suppose that $I$ contains a single initial node $g_0$ . Then, $G$ is bipartite iff it has no $ec$ nodes.
\end{lemma}
\begin{proof}
It is easy to see that nodes equidistant from the initial node $g_0$ must belong to the same partite set. A graph is bipartite iff no edge connects two such nodes, and this is the case iff $G$ has no $ec$ nodes by Definition 3.1(ii).
\end{proof}

$\;$

\noindent From Lemma 3.3 we see that nodes in graphs without ec nodes only receive a message once. For these graphs the time to termination is the number of rounds taken for the message to reach the most distant points from the set of initial nodes $I$. This is just the eccentricity of $I$, $e(I)$.
\begin{theorem}
Let $G$ be a graph and $I \subseteq G$ be a set of initial nodes with eccentricity $e(I)$. Then, flooding will have terminated after round $e(I)$ if and only if $G$ is ec-bipartite
\end{theorem}
\begin{proof}
\[
\begin{array}{llll}
\hbox{$G$ is ec-bipartite} & \hbox{iff} & \hbox{$G$ has no $ec$ nodes} & \hbox{(by Definition 3.1(iv))} \\
\; &\hbox{iff} & \hbox{no node appears in 2 round-sets} &  \hbox{(by Lemma 3.3)}  \\
\; & \hbox{iff} & \hbox{distant nodes get message last} &  \; \\
\; & \hbox{iff} & \hbox{$R_{e(I)}$ is the last non-empty round-set} &  \hbox{(by Lemma 3.2(i))}  \\
\end{array}
\]
\end{proof}
\begin{corollary}
Let $G$ be a graph and suppose that $I$ contains a single initial node $g_0$ of eccentricity $e$. Then, flooding will have terminated after round $e$ if and only if $G$ is bipartite
\end{corollary}
\begin{proof}
Follows from Lemma 3.4, Theorem 3.5 and the fact that $e = e(I)$.
\end{proof}

$\;$

\noindent To find the time to termination in general graphs we need to find a bound on when nodes can belong to a round-set for the second time. As nodes can only belong to at most two round-sets, by Theorem 2.2, this will give a bound for termination of flooding in general graphs. The following lemma relates the round-sets of second occurrences of neighbouring nodes.
\begin{lemma} Let $G$ be a graph and $I \subseteq G$ a set of initial nodes. If $h \in G$ and $h^2 \in R_j$ for some $j  \geq 1$, and if $g$ is a neighbour of $h$, then
\[
g^2 \in R_{j-1} \;\; {\it or} \;\; g^2 \in R_{j}  \;\; {\it or} \;\; g^2 \in R_{j+1}
\]
\end{lemma}
\begin{proof}
Let $i$ be the distance of $h$ from $I$, i.e. $h \in D_i$. Then, as $h^2 \in R_j$, $j>i$ by Lemma 3.2(i). As $g$ is a neighbour of $h$, $g \in D_{i}$ or $g \in D_{i-1}$ or $g \in D_{i+1}$.

\noindent {\it Case} $g \in D_i$.

\noindent As $h,g \in D_i$ are neighbours they are both $ec$ nodes. Thus, by Lemma 3.2(iii), $h^2 \in R_{i+1}$ and $g^2 \in R_{i+1}$. Therefore, $j = i+1$ and $g^2 \in R_j$.

\noindent {\it Case} $g \in D_{i-1}$.

\noindent If $g \in R_j$ ($\neq R_{i-1}\;\hbox{as} \; j>i$) then, as $g^1 \in D_{i-1} \subseteq R_{i-1}$ by Lemma 3.2(i), it must be the case that $g^2 \in R_j$. If $g\notin R_j$ and $g \in R_{j-1}$ ($\neq R_{i-1}\;\hbox{as} \; j>i$) then, as $g^1 \in   R_{i-1}$ by Lemma 3.2(i), it must be the case that $g^2 \in R_{j-1}$. If $g\notin R_j$ and $g \notin R_{j-1}$ then, as $h \in R_j$, $h$ sends to $g$ in round $j+1$ and so $g \in R_{j+1}$ ($\neq R_{i-1}\;\hbox{as} \; j>i$). As $g^1 \in R_{i-1}$, it must be the case that $g^2 \in R_{j+1}$.

\noindent {\it Case} $g \in D_{i+1}$, {\it g does not send to h in round j}.

\noindent In this case, as $h \in R_j$, $h$ sends to $g$ in round $j+1$. Thus, $g \in R_{j+1}$ ($\neq R_{i+1}\;\hbox{as} \; j>i$) and therefore, as $g^1 \in D_{i+1} \subseteq R_{i+1}$ by Lemma 3.2(i), it must be the case that $g^2 \in R_{j+1}$.

\noindent {\it Case} $g \in D_{i+1}$, {\it g sends to h in round j}.

\noindent In this case $g \in R_{j-1}$. We show that $g^1 \notin R_{j-1}$. Assume, on the contrary, that $g^1 \in R_{j-1}$. Then, by Lemma 3.2(i), $g^1 \in D_{i+1} \subseteq R_{i+1}$ and thus $j-1 = i+1$. Hence, by Lemma 3.2(i), $h^1 \in D_i \subseteq R_i = R_{j-2}$. Also, $g \notin R_{j-3}$ as $g^1 \in R_{j-1}$. To summarize:
\[
g \notin R_{j-3},\;\; h^1 \in R_{j-2},\;\; g^1 \in R_{j-1},\;\; h^2 \in R_j
\]So, $h$ sends to $g$ in round $j-1$ and $g$ sends to $h$ in round $j$ by the case assumption. This is a contradiction. Thus, the assumption that $g^1 \in R_{j-1}$ is false and, as $g\in R_{j-1}$, it follows that $g^2 \in R_{j-1}$. This completes the proof.
\end{proof}
\begin{corollary}
 Let $G$ be a graph and $I \subseteq G$ be a set of initial nodes. Then $G$ has an ec node if and only if all nodes are in exactly two round-sets. 
\end{corollary}
\begin{proof}
Follows from Lemmas 3.3 and 3.7.
\end{proof}

\noindent We can now give bounds for time to termination of flooding for graphs that have ec nodes.
\begin{theorem}Let $G$ be a graph that is not ec-bipartite with diameter $d$. Then, flooding terminates after round $j$ where $j$ is in the range 
\[
e(I)  < j \leq {\it min} \{ d(I, g_{ec}) + e(g_{ec}) +1 \; : \; g_{ec} \in G\; \hbox{{\it is an ec point}} \}.
\]
\end{theorem}
\begin{proof}
If $G$ is not ec-bipartite it has an $ec$ node $g_{ec}$, by Definition 3.1(iv).  By Lemma 3.2(iii), $g_{ec}^2 \in R_k$ where $k = d(I,g_{ec}) +1$. Let $h$ be an arbitrary node in $G$ other than $g_{ec}$. Then, there is a path
\[
h_0=g_{ec} \longrightarrow h_1 \longrightarrow \ldots \longrightarrow h_l = h
\]where $l \leq e(g_{ec})$. By repeated use of Lemma 3.7, 
\[
\begin{array}{llll}
h_1^2 \in R_{j_1} & {\it where} & k-1 \leq j_1 \leq k+1, & \; \\
h_2^2 \in R_{j_2} & {\it where} & j_1-1 \leq j_2 \leq j_1+1, & \; \\
\ldots & \; & \; & \; \\
h_l^2 \in R_{j_l} & {\it where} & j_{l-1}-1 \leq j_l \leq j_{l-1}+1. & \; \\
\end{array}
\]Thus, 
\begin{equation}
h^2_l \in R_{j_l} \;\;\; {\it where} \;\;\; k-l \leq j_l \leq k+l
\end{equation}Put $j = j_l$. From (9), as $k = d(I,g_{ec}) +1 $ and as $l \leq e(g_{ec})$,
\[
h^2_l \in R_{j} \;\;\; {\it where} \;\;\;j \leq  d(I, g_{ec}) + e(g_{ec}) +1
\]and, as $g_{ec}$ is any ec node,
\[
h^2_l \in R_{j} \;\;\; {\it where} \;\;\;j \leq {\it min}\{ d(I, g_{ec}) + e(g_{ec}) +1 \}
\]
As $G$ is not ec-bipartite, $j>e(I)$ by Theorem 3.5 and the proof is complete.
\end{proof}
\begin{corollary}Let $G$ be a graph and $I \subseteq G$ be a set of initial nodes which contain an ec node. Then, flooding terminates after round $j$ where $j$ is in the range 
\[
e(I)  < j \leq {\it min} \{ e(g_{ec}) +1 \; : \; g_{ec} \in G\; \hbox{{\it is an ec point}} \}.
\]
\end{corollary}
\begin{proof}
Follows from Theorem 3.9 as $d(I,g_{ec}) =0$ for any ec node in $I$.
\end{proof}
\begin{corollary}Let $G$ be a non-bipartite graph with diameter $d$ and a single initial node $g_0 \in G$ of eccentricity $e$. Then, flooding terminates after round $j$ where $j$ is in the range $e < j \leq e+d+1$.
\end{corollary}
\begin{proof}
Put $I = \{g_0 \}$ so that $e(I) = e$.  Then, by Theorem 3.9, flooding terminates after $j$ rounds where $j$ is in the range
\[
e  < j \leq {\it min} \{ d(g_0, g_{ec}) + e(g_{ec}) +1  :  g_{ec} \in G\; \hbox{{\it is an ec point}} \} \leq e + d + 1 .
\]
\end{proof}

\noindent Figure 2 shows that the bounds are sharp. Sharpness of the upper bound is shown by the graph on the left in which flooding terminates after $e+d+1=7$ rounds. The slightly complicated graph on the right terminates after $e+1=5$ rounds, demonstrating sharpness of the lower bound and also showing that flooding can terminate in a non-bipartite graph before round $d$ where $d$ is the diameter. In both graphs $g_0$ is the initial node.

\begin{figure}[h!]
\centering 
\includegraphics[scale=0.75]{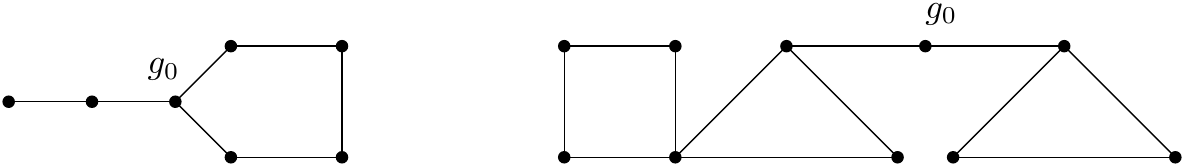} 
\caption{Sharpness of bounds}
\end{figure}

The clear separation in the termination times of bipartite and non-bipartite graphs, in the case of a single initial node, does not carry over neatly to the general case of multiple initial nodes. For the general case, it is the presence or otherwise of ec nodes that is the distinguishing property and we have defined ec-bipartite graphs for this purpose. Unlike the single node case, having ec nodes or not in the multiple initial node case cuts across the bipartite/non-bipartite boundary, i.e. bipartite graphs may have ec points and non-bipartite graphs may not. The recent work by Turau \cite{turau1} looks at termination times in terms of the `k-flooding problem' which aims to find a set S of size k such that amnesiac flooding when started concurrently by all nodes of S terminates in the least number of rounds.

\section{Dynamic flooding}
So far, we have considered flooding of a single message which initiates at a single point in time in round 0, albeit possibly from several nodes. In this section we show that amnesiac flooding will still terminate in dynamic settings where multiple floodings are initiated from multiple nodes in multiple rounds with multiple messages, even when floodings from previous initiations have not completed. We consider floodings in which nodes receiving messages in a given round only send a single message per edge in the next round. The case of `parallel flooding', where multiple floodings operate independently of each other resulting in nodes possibly sending multiple messages along an edge in a round, terminates trivially given the results in sections 2 and 3. 

There are two general cases in which messages $M_0, \ldots , M_h , \ldots $ are being flooded by nodes $x_{i_0}^0, \ldots, x_{i_h}^h, \ldots $ with initial rounds $i_0 , \ldots , i_h , \ldots$ respectively, such that at most one message per edge per round is sent. In the first case, a node that receives messages in a round sends one of the received messages $M_h$ to all neighbours from which it did not receive any message. In the second case, the node sends $M_h$ to all neighbours from which it did not receive $M_h$. We call these two cases `partial-send' and `full-send' respectively. We subdivide the second case into the subcase `ranked full-send' where the messages  $M_0, \ldots , M_h , \ldots $ are ranked and the highest ranked received message $M_h$ is sent on, and the subcase `unranked full-send' where a random message is sent on. 
\begin{definition}{\textbf{Amnesiac flooding algorithms for dynamic flooding.}}
Let $G$ be a graph, $x_{i_0}^0, \ldots , x_{i_h}^h , \ldots  \in G$ be possibly different initiating nodes flooding possibly different messages $M_0, \ldots , M_h , \ldots $ with possibly different initial rounds $i_0, \ldots , i_h , \ldots  $  respectively, where the subscript $h$ distinguishes the floodings (i.e $h \neq h' \Rightarrow (x^h_{i_h}, i_h) \neq (x^{h'}_{i_{h'}}, i_{h'})$). We assume a node $x_{i_h}$ does not initiate flooding (initiation means $x_{i_h}$ sends a message to all of its neighbours in the next round) of a message in a round in which it receives a message. Consider the properties below for an algorithm for flooding multiple messages, where $N(g)$ denotes the set of all neighbours of a node $g \in G$ and $N_{g,i,h} \subseteq N(g)$ the subset of neighbours of $g$ from which $g$ receives message $M_h$ in round $i$.
\begin{itemize}
\item[(i)]
For all $h \geq 0$, $x_{i_h}^h$ sends message $M_h$ to all of its neighbours in round $i_h+1$. 
\item[(ii)]
Suppose a node $g \in G$ receives messages from neighbours in round $i > 0$. Let $h' \geq 0$ be such that $N_{g,i,h'} \neq \emptyset $. Then, in round $i+1$, $g$ sends $M_{h'}$ to all neighbours in $N(g) \setminus \bigcup_{h \geq 0} N_{g,i,h}$  and sends no other message.
\item[(iii)]
Suppose a node $g \in G$ receives messages from neighbours in round $i > 0$ and that the $rank$ of each message $M_h$ is the subscript $h$. Let $ h'$ be the largest $h' \geq 0$ such that $N_{g,i,h'} \neq \emptyset $. Then, in round $i+1$, $g$ sends $M_{h'}$ to all neighbours in $N(g) \setminus N_{g,i,h'}$  and sends no other message.
\item[(iv)]
Suppose a node $g \in G$ receives messages from neighbours in round $i > 0$. Let $h' \geq 0$ be such that $N_{g,i,h'} \neq \emptyset $. Then, in round $i+1$, $g$ sends $M_{h'}$ to all neighbours in $N(g) \setminus N_{g,i,h'}$  and sends no other message.
\end{itemize}The flooding algorithm defined by (i) and (ii) is called {\it partial-send}, that defined by (i) and (iii) is called {\it ranked full-send} and that defined by (i) and (iv) is called {\it unranked full-send}.
\end{definition}In 4.1 and 4.2 below we show that, in the cases of partial-send and ranked full-send flooding of multiple messages, no node receives a given message more than twice, and so flooding terminates when the number of messages is finite. As a further extension to the basic flooding model of sections 2 and 3, both of these results hold in a dynamic setting where the graph $G$ may lose edges or nodes over time. In section 5.3, the general case of unranked full-send flooding of multiple messages is shown not to terminate even if the number of messages is finite.
\subsection{Multiple messages - partial-send flooding}
It may seem that flooding initiated by a node $g$ in round $i>0$  can be simulated as a round 0 flooding by adding an extra path of length $i$ to $G$ so that the message arrives at $g$ in round $i$ - in this way reducing the problem to the simple case of only round 0 floodings of sections 2 and 3. However, flooding initiations in later rounds can easily produce the same node in evenly spaced round-sets which cannot be simulated by any graph floodings from round 0 only, by the proof of Theorem 2.2.  So, the suggested simulation by a larger graph is not possible in general and an alternative proof is needed.

In the remainder of the paper, for nodes $x,y \in G$ and $h \geq 0$, $ x \rightarrow y $ will mean `$x$ sends some message to $y$', i.e. $y$ receives some message from $x$, and $x \stackrel{M_h}{\rightarrow} y$ will mean `$x$ sends message $M_h$ to $y$', i.e. $y$ receives message $M_h$ from $x$. We shall also use $\not\rightarrow$ to mean `does not send' in a similar way. 
\begin{definition}
Let $(G,E)$ be a graph with nodes $G$ and edges $E$. An {\it even flooding cycle} $c$, abbreviated {\it ef cycle}, is a mapping $c: C_{2n} \rightarrow G$ for some $n \geq 2$ where $C_{2n} = \{ 0, 1, \ldots , 2n-1 \}$, such that $\{ c(i) , c(i+1) \} \in E$ for all $i \in C_{2n}$. Arithmetic on elements of $C_{2n}$ is modulo $2n$. Whilst $C_{2n}$ is a cycle, its image $c(C_{2n})$ is a closed walk but not necessarily a cycle in $G$ - see the example in figure 3. A message $c(i) \rightarrow c(i+1)$, where $i \in C_{2n}$, in round $j \geq 0$ is a {\it clockwise} message along $c$ and a message $c(i) \rightarrow c(i-1)$, where $i \in C_{2n}$, is an {\it anticlockwise} message. 
\end{definition}
\begin{figure}[h!]
\centering 
\includegraphics[scale=0.75]{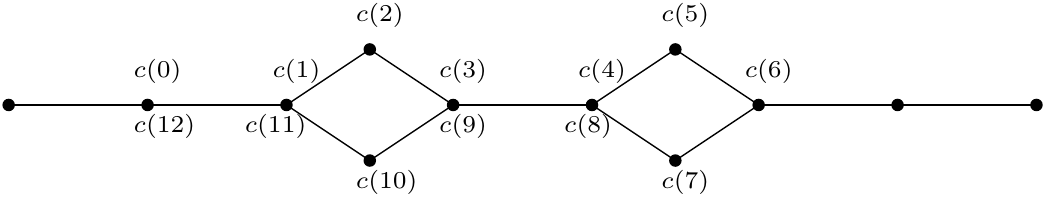} 
\caption{Even flooding cycle example $n=6$}
\end{figure}
\noindent In any given round, there will be a number of messages sent clockwise along $c$ and a number of messages anticlockwise. Our interest is in the number of messages sent in each direction from evenly spaced points of $c$. (We will refer to $i \in C_{2n}$ and $c(i) \in G$, loosely,  as `points' of $c$.)
\begin{lemma}
Let $G$ be a graph, flooding be partial-send as in Definition 4.1 and let $c$ be an ef cycle.  For all $j \geq 0$ let
\[
cl^s(even,j) = \{ i \; : \; c(i) \rightarrow c(i+1)\; in \; round \; j,\; i \; even, \; i \in C_{2n} \},
\]
\[
an^s(even,j) = \{ i \; : \; c(i) \rightarrow c(i-1)\; in \; round \; j,\; i \; even, \; i \in C_{2n} \},
\]
\[
cl^s(odd,j) = \{ i \; : \; c(i) \rightarrow c(i+1)\; in \; round \; j,\; i \; odd, \; i \in C_{2n} \},
\]
\[
an^s(odd,j) = \{ i \; : \; c(i) \rightarrow c(i-1)\; in \; round \; j,\; i \; odd,\; i \in C_{2n} \}.
\]Then, for all $j \geq0$,
\begin{equation}
| cl^s(even,j) | = |an^s(even,j)|  \;\; and \;\; | cl^s(odd,j) | = |an^s(odd,j)|.
\end{equation}
\end{lemma}
\begin{proof}
By induction on $j$. In round $j=0$, every initial node is yet to send a message and thus there are neither clockwise nor anticlockwise messages and so (10) holds trivially. Assume that (10) holds for some $j \geq 0$. It is convenient to define the corresponding sets of points of $C_n$ receiving messages  for all $j \geq 0$:
\[
cl^r(odd,j) = \{ i+1 \; : \; c(i) \rightarrow c(i+1)\; in \; round \; j,\; i \; even, \; i \in C_{2n} \},
\]
\[
an^r(odd,j) = \{ i-1 \; : \; c(i) \rightarrow c(i-1)\; in \; round \; j,\; i \; even, \; i \in C_{2n} \},
\]
\[
cl^r(even,j) = \{ i+1 \; : \; c(i) \rightarrow c(i+1)\; in \; round \; j,\; i \; odd, \; i \in C_{2n} \},
\]
\[
an^r(even,j) = \{ i-1 \; : \; c(i) \rightarrow c(i-1)\; in \; round \; j,\; i \; odd,\; i \in C_{2n} \}.
\]Clearly, 
\[
| cl^s(even,j) | = |cl^r(odd,j)|,\; |an^s(even,j)| = |an^r(odd,j)| ,
\]
\begin{equation}
 | cl^s(odd,j) |= | cl^r(even,j)| \; and \; |an^s(odd,j) = |an^r(even,j)|.
\end{equation}We show that $|cl^s(even,j+1)|=|an^s(even,j+1)|$. A clockwise message $c(i) \rightarrow c(i+1)$, where $i$ is even, is sent in round $j+1$ in the following cases:
\begin{itemize}
\item[(a)]
$c(i-1) \rightarrow c(i)$ and $c(i+1) \not\rightarrow c(i)$ in round $j$,
\item[(b)]
$c(i-1) \not\rightarrow c(i)$ and $c(i+1) \not\rightarrow c(i)$ in round $j$, but $c(i)$ was an initial node about to send to all neighbours or $c(i)$ received a message from an external neighbour, i.e. other than $c(i-1)$ or $c(i+1)$, in round $j$.
\end{itemize}The set of messages sent clockwise from $c(i)$, for even $i$, in round $j+1$ as per case (a) is $cl^r(even,j) \setminus an^r(even,j)$. We denote the set of even $i$ such that $c(i)$ sends a clockwise, equivalently anticlockwise, message in round $j+1$ as per case (b), by $ext(even,j)$. Thus,
\[
cl^s(even,j+1) = (cl^r(even,j) \setminus an^r(even, j)) \cup ext(even,j)
\]Similarly,
\[
an^s(even,j+1) = (an^r(even,j) \setminus cl^r(even, j)) \cup ext(even,j)
\]Thus,
\[
|cl^s(even,j+1)| - |an^s(even,j+1)| =|cl^r(even,j) \setminus an^r(even, j)| -  |an^r(even,j) \setminus cl^r(even, j)| 
\]
\[
=\; |cl^s(odd,j)\setminus an^s(odd,j)| - |an^s(odd,j) \setminus cl^s(odd,j)| \;\; \hbox{(by (11))} \; 
\]
\[
=\;  |cl^s(odd,j)| -|an^s(odd,j)| =0 \;\; \hbox{(by induction)}.
\]A similar inductive argument shows that $|cl^s(odd,j+1)| = |an^s(odd,j+1)|$.
\end{proof}
\begin{theorem}
Let $G$ be a graph and flooding be partial-send. Suppose that initial node $x^h_{i_h}$ floods message $M_h$ in round $i_h$. Then, if 
\begin{equation}
x^h_{i_h} \stackrel{M_h}{\rightarrow} x^h_{i_h+1} \stackrel{M_h}{\rightarrow} \ldots \stackrel{M_h}{\rightarrow} x^h_{i_h+k}
\end{equation}is a succession of sent $M_h$ messages, no node $g\in G$ can occur more than twice in (12). Thus, no flooded message is received more than twice by a node in partial-send flooding which therefore terminates if the number of floodings of messages is finite.
\end{theorem}
\begin{proof}
Assume, on the contrary, that there is a sequence of sends (12) in which some node $g \in G$ occurs at least three times. Such a sequence of sends does not occur when the flooding of $M_h$ in round $i_h$ is the only flooding in $G$ as then, by Theorem 2.2, a message cannot be received more than twice by a node. {\it If only} $M_h$ {\it is flooded}, there is $n$ with $2 \leq n < k$ such that 
\begin{equation}
x^h_{i_h} \stackrel{M_h}{\rightarrow} x^h_{i_h+1} \stackrel{M_h}{\rightarrow} \ldots \stackrel{M_h}{\rightarrow} x^h_{i_h+n}
 \stackrel{M_h}{\not\rightarrow}x^h_{i_h+n+1},
\end{equation}i.e. the succession of sends in (12) cannot proceed beyond $n$ rounds, and therefore a sequence of sends via some nodes $y_{i_h+1}, \ldots , y_{i_h+n-2} \in G$
\begin{equation}
x^h_{i_h} \stackrel{M_h}{\rightarrow} y_{i_h+1} \stackrel{M_h}{\rightarrow} \ldots \stackrel{M_h}{\rightarrow} y_{i_h+n-2}
 \stackrel{M_h}{\rightarrow} x^h_{i_h+n+1}\stackrel{M_h}{\rightarrow} x^h_{i_h+n}
\end{equation}preventing $x^h_{i_h+n}$ sending to $x^h_{i_h+n+1}$ in (13). Define an ef cycle $c: C_{2n} \rightarrow G$ by:
\begin{equation}
c(i)= x^h_{i_h+i} \; (0 \leq i \leq n), \;\; c(n+1)= x^h_{i_h+n+1}, \;\; c(2n-l) = y_{i_h+l} \;(1 \leq l \leq n-2)
\end{equation}We now consider the messages sent clockwise and anticlockwise along edges of $c$ after round $i_h$ as per the partial-send {\it flooding of all messages} in the statement of this theorem. The following notation for subsets of $C_{2n}$ will be used where $1 \leq j < j' \leq 2n-1$:
\[
[j,j']= \{ i \in C_{2n} : j \leq i \leq j' \}, \;\;(j,j')= \{ i \in C_{2n} : j < i < j' \},
\]
\[[j,j')= \{ i \in C_{2n} : j \leq i < j' \},\;\;(j,j']= \{ i \in C_{2n} : j < i \leq j' \}, \;\; [0,2n]= C_{2n}.
\]
\begin{figure}[h!]
\centering 
\includegraphics[scale=0.75]{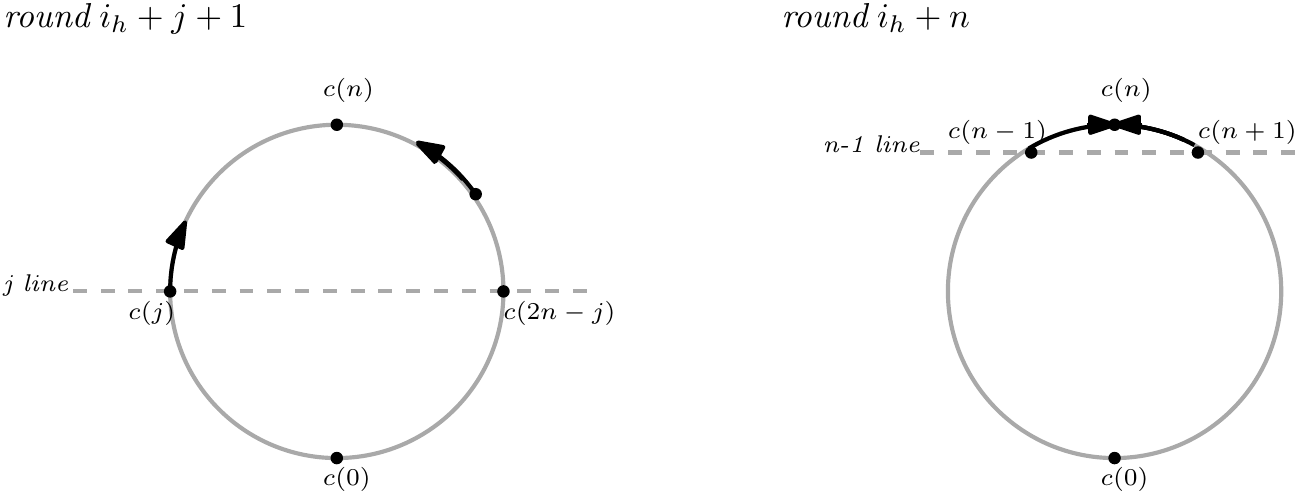} 
\caption{Proof of Theorem 4.4}
\end{figure}
\noindent We shall prove that for all $j$ with $0 \leq j \leq n-1$, the number of messages sent clockwise in round $i_h+j+1$ along $c$ above the $j$ line in the left-hand illustration in figure 4, from points of $c$ of equal parity to $j$, is less than or equal to the number of messages sent anticlockwise above the $j$ line, from points of equal parity to $j$. Precisely, if
\[
cl^s(j) = \{ i \; : \; c(i) \rightarrow c(i+1)\; in \; round \; i_h+j+1,\;\;  (i-j) \;even,\; \; i,i+1 \in [j,2n-j] \},
\]
\[
cl^r(j) = \{ i+1 \; : \; c(i) \rightarrow c(i+1)\; in \; round \; i_h+j+1,\;\;  (i-j) \;even, \;\; i,i+1 \in [j,2n-j] \},
\]
\[
an^s(j) = \{ i \; : \; c(i) \rightarrow c(i-1)\; in \; round \; i_h+j+1,\;\;  (i-j) \;even, \;\; i,i-1 \in [j,2n-j] \},
\]
\[
an^r(j) = \{ i-1 \; : \; c(i) \rightarrow c(i-1)\; in \; round \; i_h+j+1,\;\;  (i-j) \;even, \;\; i,i-1 \in [j,2n-j] \},
\]then
\begin{equation}
|cl^s(j)| (= |cl^r(j)|) \; \leq \;|an^s(j)| (= |an^r(j)|) \;\;\; (0 \leq j \leq n-1 ).
\end{equation}The proof is by induction. For $j=0$, by Lemma 4.3, 
\[
|cl^s(j)(=cl^s(even,i_h))| = |an^s(j) (= an^s(even,i_h))|.
\]Assume that (16) holds for some $j$ with $0 \leq j \leq n-2$. We show that (16) continues to hold for $j+1$. For $0\leq j \leq n-1$, define the set of points of $c$ above the $j$ line that are either initial nodes in round $i_h+j$ or receive messages from neighbours external to $c$:
\[
ext(j) = \{ i \; : \; c(i) \rightarrow c(i-1),c(i) \rightarrow c(i+1)\; in \; round \; i_h+j+1,\;\;  (i-j) \;even, \;\; i \in [j,2n-j] \}.
\]A clockwise message $c(i)\rightarrow c(i+1)$ is sent in round $i_h+j+2$ with points $i,i+1 \in [j+1,2n-(j+1)]$, i.e. $i \in [j+1,2n-(j+1))$, with $i$ of equal parity to $j+1$, in the following cases:
\begin{itemize}
\item[(a)]
$c(i-1) \rightarrow c(i)$ and $c(i+1) \not\rightarrow c(i)$ in round $i_h+j+1$. This corresponds to $i$ 
being in the set $(cl^r(j) \setminus an^r(j)) \cap [j+1, 2n-(j+1))$.
\item[(b)]
$c(i)$ is an initial node or received a message externally in round $i_h+j+1$. This means $i\in ext(j+1) \cap [j+1, 2n-(j+1))$. Note that, from (12) and (15), as $j \leq n-2 < k$, $c(j) = x^h_{i_h+j} \rightarrow x^h_{i_h+j+1}=c(j+1)$ in round $i_h+j+1$. Thus, $c(j+1) \not\rightarrow c(j)$ in round $i_h+j+2$ and so $j+1 \notin ext(j+1)$.  Hence, $i \in ext(j+1) \cap (j+1, 2n-(j+1))$.
\end{itemize}From (a) and (b),
\begin{equation}
cl^s(j+1)=(cl^r(j) \setminus an^r(j))\cap [j+1,2n-(j+1)) \cup ext(j+1) \cap (j+1,2n-(j+1)).
\end{equation}
An anticlockwise message $c(i)\rightarrow c(i-1)$ is sent in round $i_h+j+2$ with points $i,i-1 \in [j+1,2n-(j+1)]$, i.e. $i \in (j+1,2n-(j+1)]$, with $i$ of equal parity to $j+1$, in the following cases:
\begin{itemize}
\item[(a')]
$c(i+1) \rightarrow c(i)$ and $c(i-1) \not\rightarrow c(i)$ in round $i_h+j+1$. This corresponds to $i$ 
being in the set $(an^r(j) \setminus cl^r(j)) \cap (j+1, 2n-(j+1)]$. 
Note that, from (12) and (15), as $j \leq n-2 < k$, $c(j+1) = x^h_{i_h+j+1} \rightarrow x^h_{i_h+j+2}=c(j+2)$ in round $i_h+j+2$. Thus, $c(j+2) \not\rightarrow c(j+1)$ in round $i_h+j+1$ and so $j+1 \notin an^r(j)$. Hence, $i \in  (an^r(j) \setminus cl^r(j)) \cap (j+1, 2n-(j+1)] = (an^r(j) \setminus cl^r(j)) \cap [j+1, 2n-(j+1)]$.
\item[(b')]
$c(i)$ is an initial node or received a message externally in round $i_h+j+1$. This means $i\in ext(j+1) \cap (j+1, 2n-(j+1)]$. 
\end{itemize}
From (a') and (b'),
\begin{equation}
an^s(j+1)=(an^r(j) \setminus cl^r(j))\cap [j+1,2n-(j+1)] \cup ext(j+1) \cap (j+1,2n-(j+1)].
\end{equation}
Now, $ext(j+1) \cap cl^r(j) = \emptyset$ and $ext(j+1) \cap an^r(j) = \emptyset$ as points in $ext(j+1)$ send messages in both clockwise and anticlockwise directions along $c$ in round $i_h+j+2$ and therefore cannot have received messages from either direction in round $i_h+j+1$ and so cannot belong to either $cl^r(j)$ or $an^r(j)$. Therefore, from (17) and (18), to prove the inductive step $|cl^s(j+1)| \leq |an^s(j+1)|$, it suffices to prove that
\begin{equation}
|(cl^r(j) \setminus an^r(j))\cap [j+1,2n-(j+1)) | \leq |(an^r(j) \setminus cl^r(j))\cap [j+1,2n-(j+1)]|
\end{equation}By definition, $cl^r(j),an^r(j) \subseteq [j,2n-j]$ are points of $c$ of opposite parity to $j$. Thus, $cl^r(j),an^r(j) \subseteq [j+1,2n-(j+1)]$, $(cl^r(j) \setminus an^r(j))\cap [j+1,2n-(j+1)) = (cl^r(j)\setminus \{ 2n-(j+1) \}) \setminus an^r(j)$ and $(an^r(j) \setminus cl^r(j))\cap [j+1,2n-(j+1)] = an^r(j)\setminus cl^r(j)$, and so (19) follows as $| cl^r(j)| \leq | an^r(j)|$ by induction. Thus, (16) holds for all $j$ with $0 \leq j \leq n-1$.

Put $j= n-1$ in (16). Then, $cl^s(n-1) \subseteq [n-1, n+1 ]$, $c(n-1)=x^h_{i_h+(n-1)} \rightarrow x^h_{i_h+n}=c(n)$ in round $i_h+n$ by (15) and (12), and so $cl^s(n-1)=n-1$. As $|cl^s(n-1)| \leq |an^s(n-1)|$, $an^s(n-1)= n+1$. Thus, $c(n+1) \rightarrow c(n)$, i.e. (by (15)) $x^h_{i_h+n+1} \rightarrow x^h_{i_h+n}$ in round $i_h+n$. (See the right-hand illustration in figure 4.) This contradicts (12) in which $x^h_{i_h+n} \rightarrow x^h_{i_h+n+1}$ in round $i_h+n+1$. Thus, the assumption that (12) contains some node more than twice is incorrect.
\end{proof}


An example of multiple messages being flooded by a partial-send algorithm is the flooding of the details and price of an item that a node purchases. For example, a node receiving the details and prices of purchases of the item from its neighbour nodes might select the cheapest and pass on the information of the cheapest available item to neighbours from which it did not receive messages about the item in the previous round. There would be no reason to send a message to a neighbour from which details of a more expensive item had been received in the previous round as that neighbour would already have made the purchase - hence a partial-send flooding algorithm would be used.

The special case of Theorem 4.4, where the same message is flooded by all initial nodes, generalizes multi-source flooding of sections 2 and 3 to the case where sources may initiate floodings in different rounds.
\begin{corollary}
Flooding of a single message from multiple sources initiating in possibly different rounds always terminates.
\end{corollary}


\noindent Corollary 4.5 alongside the duality with `reverse' flooding can be used to show termination of algorithms such as leader election \cite{PelegLE-JPDC-90} which may have deviations from the strict flooding algorithm. For example,  a `sink' node, which receives messages from all neighbours in a certain round $i$, reverse floods a message to all neighbours in round $i+1$ as part of the algorithm.  This might occur when (forward) flooding has not completed in other parts of the graph. It is not clear that this abnormal behaviour at such a sink node is not out of step and cannot combine with other uncompleted flooding to cause non-termination. (Note that the node cannot be considered an initial node in round $i+1$ as it received messages in round $i$.) 
\begin{corollary}
Let $(G,E)$ be a graph which is flooded with a single message from round 1  as in Definition 1.1. Suppose that node $g \in G$ is a {\it sink} node in round $i$, i.e. $g' \rightarrow g$ for all neighbours $g'$ of $g$. Further, suppose that flooding proceeds according to the algorithm in Definition 1.1  except that in round $i+1$ $g$ sends to all of its neighbours. Then, the resulting flooding process terminates.
\end{corollary}
\begin{proof}
The {\it state} of flooding $S_i$ in any round $i\geq 1$ which determines the subsequent states of flooding is defined by a relation making $G$ into a bidirected graph \cite{Edm03}  where edges may be undirected, directed one way or directed both ways as in our illustrations of rounds in figures:
\[
S_i \subseteq G \times G , \;\; S_i = \{ (g_1,g_2) \; : \; \{ g_1 ,g_2 \} \in E, g_1 \rightarrow g_2  \; in \; round \; i \}
\]Let the normal flooding process, where $g$ adheres strictly to the flooding algorithm in round $i+1$, terminate after round $\tau$ and let the sequence of states from round 1 be:
\begin{equation}
S_1, \ldots , S_i , S_{i+1} , \ldots , S_{\tau}.
\end{equation}We need to show that flooding terminates from the state $S_{i+1}^+$ given by:
\begin{equation}
S_{i+1}^+ = S_{i+1} \cup \{ (g,g') \; : \; g \in N_g \}
\end{equation}Consider the {\it reverse} flooding of (20), in which received messages $g_1 \rightarrow g_2$ become sent messages $g_2 \rightarrow g_1$ and sink nodes become initial nodes, i.e. given by the sequence of states:
\begin{equation}
inv(S_{\tau}), \ldots , inv(S_{i+1}) , inv(S_i), \ldots , inv(S_1)
\end{equation}where, for any state $S$, $inv(S)$ is the inverse relation $\{ (g_2,g_1) : (g_1,g_2) \in S \}$. The node $g$ is an initial node in $inv(S_{i})$ as $g$ is a sink node in state $S_{i}$ in forward flooding. Let
\begin{equation}
inv(S_i)^- = inv(S_i) \setminus \{ (g,g') \; : \; g' \in N_g \} 
\end{equation}The following
\[
inv(S_{\tau}), \ldots , inv(S_{i+1}) , inv(S_i)^-
\]are successive states of flooding as in (22) except that $g$ is not an initial node in $inv(S_i)^-$. By Corollary 4.5, this flooding terminates say after additional succeeding states $S_1' , \ldots , S_{\tau}'$:
\[
inv(S_{\tau}), \ldots , inv(S_{i+1}) , inv(S_i)^-, S_1' , \ldots , S_{\tau'}'
\]By duality, the sequence
\begin{equation}
inv(S_{\tau'}'), \ldots , inv(S_1' ), inv( inv(S_i)^-), S_{i+1}, \ldots , S_{\tau}
\end{equation}terminates at $S_{\tau}$. Now, by (23), $\{(g',g) : g' \in N_g \} \cap  inv( inv(S_i)^-) = \emptyset$ and so $g$ does not receive any messages in state $ inv( inv(S_i)^-)$. Thus, by Corollary 4.5, the flooding sequence of states obtained from (24), where $g$ is an initial node in state $ S_{i+1}^+$ as in (21),
\[
inv(S_{\tau'}'), \ldots , inv(S_1' ), inv( inv(S_i)^-), S_{i+1}^+ ,
\]will terminate from the state $ S_{i+1}^+$ (after some additional  succeeding states).
\end{proof}

\subsection{Multiple messages - ranked full-send flooding}
Here, messages $M_0, \ldots , M_h , \ldots $ are flooded as per Definition 4.1 conditions (i) and (iii). The rank of message $M_h$ is the integer $h$. Messages are ranked according to the round they were flooded, and messages flooded in the same round by different nodes also have different ranks, i.e.
\begin{equation}
h_1  \leq  h_2 \Rightarrow i_{h_1} \leq i_{h_2} \;\;\; (h_1,h_2 \geq 0)
\end{equation} Note that, unlike in partial-send flooding, a node $g$ may send a message to a node from which it has just received a message. However, the message sent will be different (of higher rank) to the one received (see figure 5). 
\begin{figure}[h!]
\centering 
\includegraphics[scale=0.75]{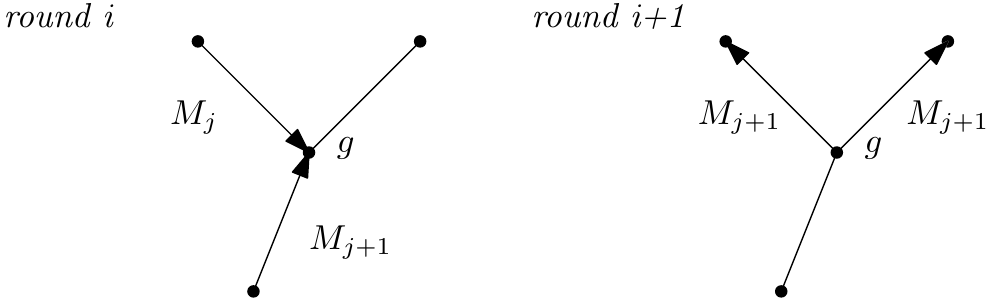} 
\caption{Ranked full-send flooding}
\end{figure}We will prove that a message cannot be received by a node more than twice in ranked full-send flooding by defining round-sets with respect to the different messages.
\begin{definition}Let $G$ be a graph flooded by ranked full-send flooding with initiating nodes and corresponding rounds denoted as in Definition 4.1. For each of these messages $M_h$, the {\it round-sets} $R_{i_h}^{M_h}, R_{i_h+1}^{M_h}, \ldots $ of $M_h$ are defined as:
\[
\begin{array}{ll}
R_{i_h}^{M_h}  & \hbox{\it is the set of initial nodes } x_{ i_{h'}} \; such \; that \; i_{h'}=i_h \; and\; h' \geq h,  \\
R_{i_h+j}^{M_h}  & \hbox{\it is the set of initial nodes or that receive message $M_{h'}$, where $ h' \geq h$,  in round $i_h+j$}  \; (j > 0). \\
\end{array}
\] 
\end{definition}\noindent So, round-sets are defined for each message $M_h$ after it is flooded. They comprise nodes that receive the message or a higher ranked message.
\begin{theorem}
Given any flooded message $M_h$ where $h \geq 0$, any node $g \in G$ receives $M_h$ at most twice. Thus, ranked full-send flooding terminates if the number of messages flooded is finite.
\end{theorem}
\begin{proof}
Assume that a node $g$ receives a message $M_h$ more than twice. If $g$ receives $M_h$ three or more times, then there are evenly spaced rounds $s$ and $s+d$, where $d$ is even, in which $g$ receives $M_h$. By Definition 4.7,  $g$ is in the corresponding round-sets of $M_h$ for those rounds. So, if ${\mathcal R^e}$ is the set of finite sequences of consecutive round-sets satisfying:
\begin{equation}
\underline{R} = 
R_s^{M_h}, \ldots, R_{s+d}^{M_h} \;\;\; \hbox{ {\it where} $s\geq i_h, \;d>0$}, \; d \; even  
\end{equation}and
\begin{equation}
\hbox{{\it there exists} $g \in R_s^{M_h} \cap R_{s+d}^{M_h} $  {\it such that} $g$ {\it receives message} $M_h$ {\it in round} $s+d$},
\end{equation}then ${\mathcal R^e}$ is non-empty. (For the purposes of the proof, we only require that $g \in R_s^{M_h}$ in round $s$ and do not require that $g$ receives $M_h$ in that round.)  Define the subset ${\mathcal R}^e_{\hat{d}}$ of ${\mathcal R}^e$ of sequences of minimum (even) duration $\hat{d}$ and $\underline{R}^* $ in ${\mathcal R}^e_{\hat{d}}$ with earliest start point $\hat{s}$. By (27), there exists $g \in R_{\hat{s}}^{M_h} \cap R_{\hat{s}+\hat{d}}^{M_h}$ which receives message $M_h$ in round $\hat{s}+\hat{d}$. Choose node $g'$ which sends message $M_h$ to $g$ in round $\hat{s}+\hat{d}$. Observe that, as $g'$ sends $M_h$ in round $\hat{s}+\hat{d}$ and $i_h \leq \hat{s} < \hat{s}+\hat{d}-1$ and so $M_h$ was initially flooded earlier than round $\hat{s}+\hat{d}-1$, 
\begin{equation}
g' \; {\it receives}\; M_h \; {\it in} \; {\it round} \; \hat{s}+\hat{d}-1.
\end{equation}Now, either $g$ does not send a message to $g'$ in round $\hat{s}+1$ or $g$ does send a message to $g'$ in round $\hat{s}+1$. We show that each of these cases leads to a contradiction.  

$\;$

\noindent {\it Case (i) $g$ does not send a message to $g'$ in round $\hat{s}+1$}

\noindent As $g'$ is a neighbour of $g$, $g'$ sends a message to $g$ in round $\hat{s}$. Now, $g'$ does not send a lower ranked message than $M_h$ to $g$ in round $\hat{s}$ - otherwise $g$ which does receive a higher ranked message than $M_h$ in round $\hat{s}$ from some neighbour would, by the flooding algorithm in Definition 4.1(iii), send a higher ranked message to $g'$ in round $\hat{s}+1$ contradicting the case assumption that $g$ does not send any message to $g'$ in round $\hat{s}$. Thus, $g'$ sends a message $M_{h'}$, where $h' \geq h$, to $g$ in round $\hat{s}$ and so $g'$ is an initial node flooding $M_{h'}$ or it receives $M_{h'}$ in round $\hat{s}-1$. Hence, $\hat{s}-1 \geq i_{h'} \geq i_h$ by (25) as $h ' \geq h$. Therefore, by Definition 4.7, $g' \in R_{\hat{s}-1}^{M_h}$. Consider the sequence 
\begin{equation}
\underline{R}^{*'} = R_{\hat{s}-1}^{M_h},R_{\hat{s}}^{M_h},\ldots, R_{\hat{s}+\hat{d}-1}^{M_h} 
\end{equation}By (28), we have that $g'$ receives $M_h$ in round $\hat{s}+\hat{d}-1$ and so $ g' \in R_{\hat{s}-1}^{M_h} \cap R_{\hat{s}+\hat{d}-1}^{M_h}$. Also, $i_h \leq \hat{s}-1$. Thus, $\underline{R}^{*'}$ at (29) satisfies (26) and (27) and so $\underline{R}^{*'} \in {\mathcal R}^e$. Its duration $\hat{d}$  is the same as that of  $\underline{R}^*$ but its start point  $\hat{s}-1$  is earlier, contrary to our choice of $\underline{R}^* $ as having the earliest start point.

$\;$

\noindent {\it Case (ii) $g$ sends a message to $g'$ in round $\hat{s}+1$}

\noindent Here, as $g \in R_{\hat{s}}^{M_h}$, it follows that $g' \in R_{\hat{s}+1}^{M_h}$ and so $g' \in R_{\hat{s}+1}^{M_h} \cap R_{\hat{s}+\hat{d}-1}^{M_h}$. If $\hat{d}=2$, then $\hat{s}+1 = \hat{s}+\hat{d}-1$ and   
\[
\underline{R}^{*} =  R_{\hat{s}}^{M_h}, R_{\hat{s}+\hat{d}-1}^{M_h}, R_{\hat{s}+\hat{d}}^{M_h}
\]As $g' \in R_{\hat{s}+1}^{M_h} = R_{\hat{s}+\hat{d}-1}^{M_h}$, $g'$ receives $M_h$ or a newer message $M_h$, by Definition 4.1(iii), in round $\hat{s}+\hat{d}-1$. Also, $g'$ sends $M_h$ in round $\hat{s}+\hat{d}$. Thus, $g'$ cannot  receive $M_h$ in round $\hat{s}+\hat{d}-1$ as that would contradict basic flooding. Neither can $g'$ receive a newer $M_{h'}$ in round $\hat{s}+\hat{d}-1$ as that would mean that $M_h$ would not be sent in round $\hat{s}+\hat{d}$ by  Definition 4.1(iii). It follows that $\hat{d}$ cannot be 2. So, we have the sequence
\[
\underline{R}^{*''} =  R_{\hat{s}+1}^{M_h},\ldots, R_{\hat{s}+\hat{d}-1}^{M_h}
\]As $g' \in R_{\hat{s}+1}^{M_h} \cap R_{\hat{s}+\hat{d}-1}^{M_h}$ and, by (28), $g'$ receives $M_h$ in round $\hat{s}+\hat{d}-1$, (26) and (27) are satisfied by $\underline{R}^{*''}$ and so $\underline{R}^{*''}\in {\mathcal R}^e$.  The duration of $\underline{R}^{*''}$,  $\hat{d} -2$, is smaller than the duration $\hat{d}$ of $\underline{R}^{*}$ which was chosen to have the smallest even duration. This is the required contradiction and completes the proof of the theorem.
\end{proof}


Ranked full-send flooding may be used where time-sensitive information is being rapidly disseminated from various sources. For example, the prices of stocks and shares. A received node may act on such timely information continuously in an unspecified way
and will send the newest information from the most reliable sources to all neighbours who did not forward that information. Note that the unrefined ranked full-send algorithm may result in a node receiving older information from a more distant mode later on. This may or may not present a problem depending on the way a node acts on received information.

Uses of the flooding algorithm for broadcasting will need to ensure that every message is received by any given node in some round. 
It is easy to see that (time-)ranked full-send flooding achieves this if the same single node floods the sequence of messages $M_0, \ldots , M_h, \ldots$ to broadcast a stream of messages. If a newly flooded message $M_{i_{h'}}$ is received by a node $g \in G$ in some round $i$ for the first time, this will be before any newer message is received by $g$, and therefore $g$ will send $M_{i_{h'}}$, in round $i+1$, to all neighbours of $g$ which did not send $M_{i_{h'}}$ to $g$ in round $i$. From this, it is clear that all neighbours of $g$ will have received $M_{i_{h'}}$ by round $i+1$. 
\begin{corollary}
If (in ranked full-send flooding) a single node floods messages $M_0, \ldots , M_h , \ldots $ with initial rounds $i_0 < \ldots < i_h < \ldots$ respectively, then every node receives every message and no node receives a given message more than twice.
\end{corollary}


If different nodes flood different messages, it cannot be guaranteed that every node receives every message in either partial-send or full-send flooding. One solution is for different messages received by a node, to be sent to neighbours in different successive rounds. This would be needed if message size is restricted as in a CONGEST model \cite{pelegbook} of message passing. The flooding process from the point of view of an observer of the progress of a given flooding of a message $M_h$, would appear as delays of more than one round between when $M_h$ is received and sent by a node. Such delays could lead to non-termination of the flooding of $M_h$, as is discussed in subsections 5.1 and 5.2 below. Another solution, if message size is unrestricted as in the LOCAL model \cite{pelegbook} of message passing, is basic flooding of all messages as in sections 2 and 3, independently and in parallel. This can result in growing message sizes as more and more messages of infinitely many would be flooded. A better solution is a parallel flooding version of ranked full-send flooding where a node which receives both older and some newer messages initially flooded by the same node, only sends the newer message. This reduces message size and all older and newer messages broadcast by any node are still received by every node. This follows by Corollary 4.9 applied to the observer of the sequence of floodings broadcast by a given node, which would effectively be ranked full-send flooding emanating from that node. It also means that, for the infinite sequences of messages flooded over time by a set $B\subseteq G$ of broadcasting nodes, at most $|B|$ messages are sent along any link in any round, thus requiring only a fixed finite size of message for transmission.  Moreover, no record of previous rounds is kept. 

\subsection{Loss of edges or nodes during flooding}
In the termination results Theorems 2.2, 4.4 and 4.8, we assume that the graph is fixed for the duration of the flooding process. Flooding proceeds in rounds acting on a fixed graph $(G,E)$ in all rounds, where $G$ is a fixed set of nodes and $E$ a fixed set of edges. The proofs in Theorems 2.2, 4.4 and 4.8 are unaffected if we allow the possibility of edge loss as flooding progresses, i.e. flooding proceeds in rounds on respective graphs:
\begin{equation}
(G,E_0), (G,E_1), \ldots , (G, E_i) , (G, E_{i+1}), \ldots  \;\; {\it where} \;\; E_0 \subseteq E_1 \subseteq \ldots \subseteq  E_i \subseteq E_{i+1} \subseteq \ldots
\end{equation}Allowing loss of nodes follows easily. Given a graph $(G,E_i)$ and a node $g \in G_i$, let $E_i^g$ be the edges in $E_i$ not incident with $g$. Then, $(G\setminus \{g \},E_i^g)$ is the subgraph of $(G,E_i)$ induced by the subset of nodes $G\setminus \{g \}$ and $(G, E_i^g)$ is the subgraph of  $(G,E_i)$ in which $g$ is isolated, i.e. with the same nodes as $(G,E_i)$ but with edges incident with $g$ removed. Clearly, given the sequence of graphs (30), flooding terminates on the sequence of graphs with monotonically decreasing sets of edges:
\[
(G,E_0), (G,E_1), \ldots , (G, E_i^g) , (G,E_{i+1}^g), \ldots
\]if and only if it terminates on the sequence of graphs where $g$ becomes isolated: 
\[
(G,E_0), (G,E_1), \ldots , (G\setminus \{ g \}, E_i^g) , (G \setminus \{g \} ,E_{i+1}^g), \ldots
\]It follows, inductively, that flooding terminates on any sequence of graphs of the form:
\[
(G_0,E_0), (G_1,E_1), \ldots , (G_i, E_i) , (G_{i+1}, E_{i+1}), \ldots  \;\; {\it where} \;\; G_i \subseteq G_{i+1}, E_i \subseteq E_{i+1} \; \hbox{{\it for all i}} \geq 0.
\]

\section{Non-terminating cases}
In this section, cases of amnesiac floodings that do not terminate generally are given.

\subsection{Asynchronous flooding}

We consider firstly an asynchronous message passing model in which an adaptive adversary can cause flooding to be non-terminating. In the model, the computation still proceeds in global synchronous rounds but the adversary can decide the delay of message delivery on any link. A message cannot be lost and will eventually be delivered, but the adversary can decide in which round to deliver the message. The adaptive adversary can decide on individual link delays for a round, based on the state of the network for the present and previous rounds (i.e. node states, messages in transit and message history).  Otherwise, the flooding algorithm is the same as for the synchronous case. We show that the adversary can force asynchronous flooding to be non-terminating by adaptively choosing link delays. Consider the triangle graph in figure 6. In round 0, $b$ has a message $M$ as the source node which is then flooded. Figure 6 illustrates rounds 1-5. Both $a$ and $c$ send $M$ to each other in round 2. In round 3,  $a$ sends $M$ to $b$ but the adversary makes $c$ hold the message for one round, indicated  by a circled node. After that, in round 4, $b$ and $c$ send $M$ to each other. In any given round, the corresponding bidirected graph determines in which ways flooding can proceed. As the bidirected graphs in rounds 2 and 4 are equivalent (by a relation- and adversary-preserving isomorphism), flooding can proceed from round 4 in similar ways as from round 2 with the edges $b$ and $a$ interchanged and the adversary acting on node $c$ as before. The repeated adversarial intervention can prevent flooding from terminating.
\begin{figure}[h!]
\centering 
\includegraphics[scale=0.75]{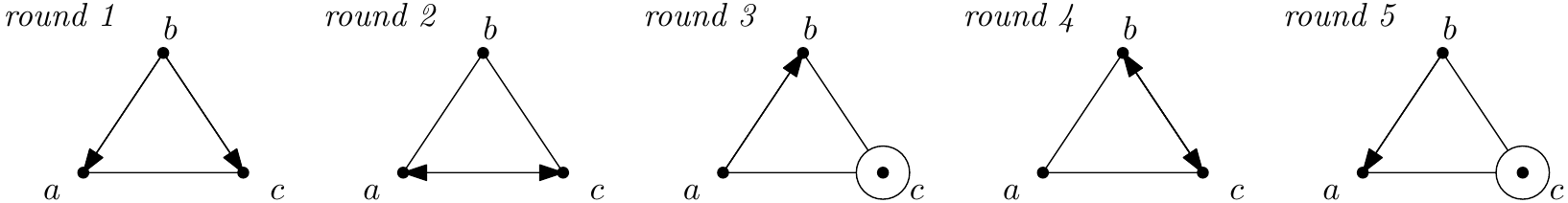} 
\caption{Non-termination of asynchronous flooding}
\end{figure}

\subsection{Fixed delays at edges}
In the asynchronous example in figure 6 above, the delay was caused by an adversary at node $c$. The duration of message transmission along the edges $\{ b, c \}$ and $\{ a , c \}$ depended on whether the adversary $c$ was receiving or sending the message. If the delay along any edge is fixed and the same in both directions on the triangle graph, then flooding terminates. This raises the question of whether flooding terminates on any graph if delays along edges are the same in both directions. Consider the example in figure 7 below. We have a weighted graph with edges labelled by positive integers representing the message transit time in either direction.
\begin{figure}[h!]
\centering 
\includegraphics[scale=0.75]{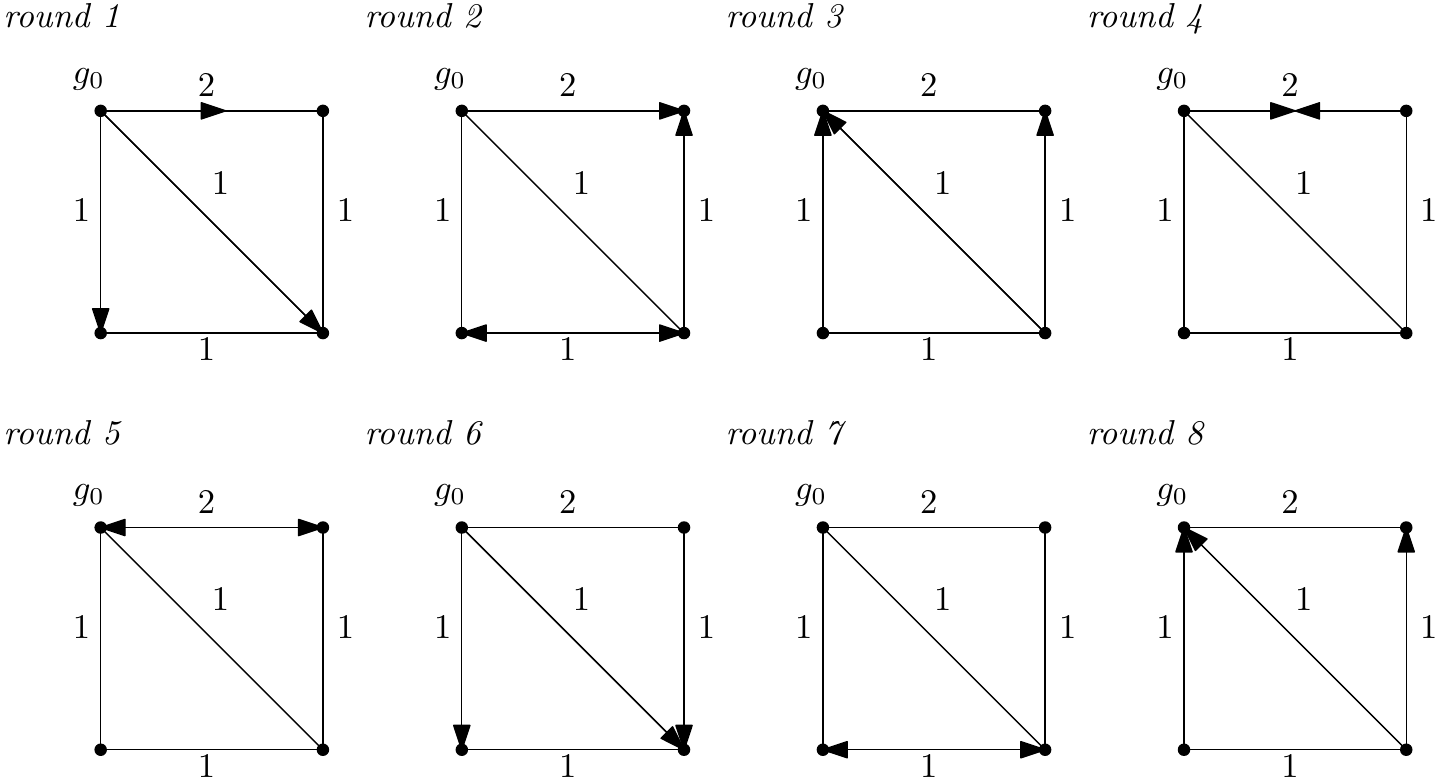} 
\caption{Non-termination of flooding with fixed edge delays}
\end{figure}Flooding proceeds with respect to a global clock, the non-negative integer ticks of which define rounds. Flooding is initiated in round 1 when node $g_0$ sends to all neighbours. Message transit time is given in terms of the number of ticks of the clock taken. Arrowheads at nodes indicate a message being received at a node. Arrowheads not at nodes (as in round 4 of figure 7) indicate messages in transit. The arrowheaded graph in a round determines how flooding proceeds. As the arrowheaded graphs in rounds 3 and 8 are identical, flooding does not terminate.

\subsection{Multiple messages - unranked full-send flooding}
If a flooding algorithm for multiple messages chooses to send a received message to all neighbours from which it did not receive that message, and another node receiving the same messages chooses a different message then that can lead to non-termination. In figure 8, $M_0$ is flooded by $x_0$ in round 1 and $M_1$ by $x_1$ in round 2. Nodes $g$ and $g'$ differ in the message they choose to forward in round 3. As a result the very same messages will be sent in round 9 and non-termination will ensue.
\begin{figure}[h!]
\centering 
\includegraphics[scale=0.75]{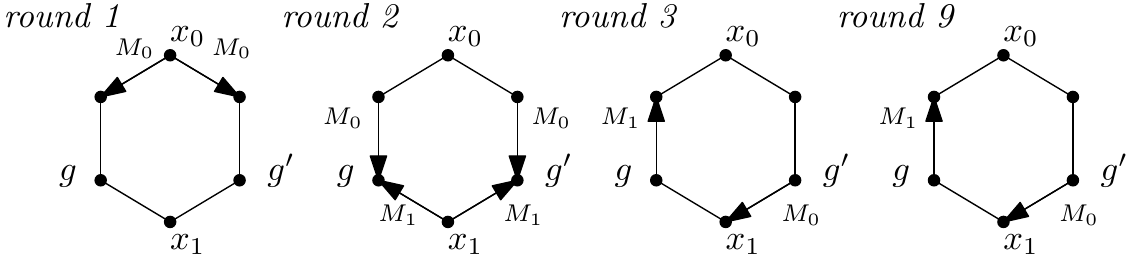} 
\caption{Non-termination of unranked full-send flooding}
\end{figure}


\subsection{Addition of edges or nodes during flooding}
In the simple example in figure 9 below, node $b$ initiates flooding in round 1. Edge $\{ b , c \}$ is added in round 2. The message is sent to from $c$ to $b$ along that edge in round 3 and circulates in the graph forever. In the case of graphs that are odd cycles, if ranked full-send flooding is continuous, a later higher-ranked flooded message will eventually eliminate such lower-ranked rogue messages. However, if the graph has even cycles, it is possible for such a lower-ranked message to pass higher-ranked messages and evade elimination infinitely often.
\begin{figure}[h!]
\centering 
\includegraphics[scale=0.75]{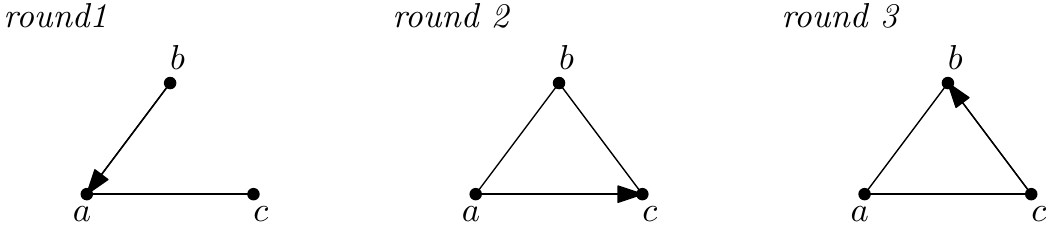} 
\caption{Non-termination of flooding with addition of edges}
\end{figure}


\section{Conclusion}
In this paper we have shown that synchronous flooding terminates on any finite graph and that no node receives the message more than twice. From this we derived sharp bounds for termination times. We have extended to dynamic settings of potential use for applications such as broadcast and leader election algorithms, with multiple initiation times and messages.  Our proofs showed that sequences of flooded messages originating from a given source could not survive longer in the presence of other dynamically flooded messages than if flooded alone from that source, and so termination and time to termination results carry over from the basic non-dynamic case. For ranked full-send flooding, we assumed that message ranking respected the order of time of initiation. This has a simpler proof. It is possible that some applications may rank messages otherwise. In fact, Theorem 4.8 can be proved without this assumption in a similar way to Theorem 4.4 using ef cycles. 

Multiple-message flooding could be extended to the case where a given node floods different messages to different neighbours simultaneously when initiating a flooding. This would compound the problem of delivering every message to every node, if broadcast is the objective. This problem can be solved in network topologies which have multiple edge-disjoint Hamilton cycles - common examples include complete graphs,  n-cubes \cite{OkS00}, tori \cite{bb03}, star networks \cite{cada}, \cite{DeH18} and transposition graphs \cite{Hus16} -  or in graphs that may not be hamiltonian but may have edge-disjoint circuits \cite{eul92} which visit every node in which nodes but not edges may be repeated. Multiple messages are sent in a given direction along different edge-disjoint circuits. Messages are forwarded by a node along the circuit on which they arrived. This guarantees that all messages sent by all nodes are received. However, the broadcast of a message takes at least as many rounds as there are vertices in the graph, unlike flooding which takes no more than about twice the diameter number of rounds. To improve the time taken to broadcast messages, when there are fewer such messages circulating in the network, a receiving node could forward a message received by a circuit edge further along that circuit, but also send the message to other unused edges in a `flooding' fashion to hasten the broadcast. Such a flooding algorithm would be amnesiac, as the sending of messages would be determined only by the received messages in the previous round. Variations of such algorithms on different topologies can be simulated on a small scale and properties compared using temporal logic model checkers \cite{raed}.

Another possible area of investigation is the design of  `functional' self-healing algorithms~\cite{CDT2018-CompactFTZ,CastanedaLT-Compact-ICDCN2020}   such that an ongoing flooding will still terminate despite nodes being adversarially deleted or inserted with judicious insertion of edges and maintaining of other invariants such as connectivity and expansion\cite{PanduranganRT16-DEX,PanduranganXhealT14,FG-DCJournal2012,Amitabh-2010-PhdThesis,HayesPODC08,SaiaTrehanIPDPS08}.






\end{document}